\documentclass[11pt]{article}
\usepackage{fullpage}

\usepackage{amsthm}
\usepackage{amsmath}
\usepackage{amsfonts}

\usepackage{multirow}
\usepackage{bigstrut}
\usepackage{xspace}
\usepackage{graphicx}
\usepackage{comment}
\usepackage{cite}
\usepackage{color}

\usepackage[boxed,linesnumbered]{algorithm2e}

\incmargin{1em}

\newtheorem{theorem}{Theorem}
\newtheorem{lemma}[theorem]{Lemma}

\newcommand{\remove}[1]{}
\newcommand{\ignore}[1]{}

\providecommand{\ceil}[1]{\left\lceil#1\right\rceil}

\newcommand{\ceilsqrt}[1]{\left\lceil\sqrt{#1}\,\right\rceil}

\newcommand{\concept}[1]{\textbf{#1}}

\providecommand{\ceil}[1]{\left\lceil#1\right\rceil}

\SetKwInput{LocalData}{local data}
\SetKwInput{SharedData}{shared data}
\SetKwInput{Constants}{constants}
\SetKw{Break}{break}
\SetKwData{MyID}{pid}
\SetKw{True}{true}
\SetKw{False}{false}
\SetKw{KwAnd}{and}
\SetKw{Not}{not}

\SetKw{Right}{right}
\SetKw{Down}{down}
\SetKw{Stop}{stop}

\begin{document}

\title{Slightly smaller splitter networks}

\author{James Aspnes\thanks{Supported in part by NSF grant
CCF-0916389.}\\Yale University}

\maketitle

\begin{abstract}
The classic renaming protocol of Moir and Anderson~\cite{MoirA1995}
uses a network of $\Theta(n^2)$ splitters to assign unique names to
$n$ processes with unbounded initial names.
We show how to reduce this bound to $\Theta(n^{3/2})$ splitters.
\end{abstract}

\section{Introduction}
\label{section-introduction}

We show how to reduce the $\Theta(n^2)$ space and output
namespace of renaming using a splitter network in the style of Moir
and Anderson~\cite{MoirA1995} to $\Theta(n^{3/2})$.  
The individual time complexity remains $\Theta(n)$, which is optimal for deterministic
renaming given an unbounded initial namespace~\cite{ChlebusK2008}.

Our construction is based on alternating small Moir-Anderson grids
with layers of small binary trees.
The resulting
renaming algorithm is not even remotely competitive with the
tight output namespace and polylogarithmic time complexity of the best
currently known randomized renaming
algorithm~\cite{AlistarhAGGG2010}, and requires more
space, more time, and a larger output
namespace than the best currently known deterministic
algorithm~\cite{ChlebusK2008} in the case where 
the initial names are
sub-exponentially large.
However, it uses less space than any other currently known algorithm
when the initial names are unbounded, and might perhaps be useful as
an initial stage before a better algorithm under such circumstances.

\subsection{Splitter networks}
\label{section-splitter-networks}

A \concept{splitter}~\cite{MoirA1995} is a shared-memory object,
implemented from two multi-writer atomic registers, with a
single operation that returns a value $\Right$, $\Down$,
or $\Stop$.  Splitters satisfy the following conditions:

\begin{itemize}
\item In any execution of a splitter, at most one process obtains the
value $\Stop$.
\item If only one process invokes a splitter, that process obtains
$\Stop$.
\item If at least two processes invoke a splitter, at least one
process obtains either $\Stop$ or $\Right$ and at least one process
obtains either $\Stop$ or $\Down$.
\end{itemize}

We can think of splitters as routing components of a network, where
the $\Right$ and $\Down$ outputs send processes along virtual
``wires'' to further splitters.
Figure~\ref{figure-moir-anderson} shows the splitter network used by
Moir and Anderson~\cite{MoirA1995}.
It consists of two-dimensional triangular grid of
$\binom{m}{2}$ 
splitters, containing splitters at all positions $(i,j)$ where $0 \le
i, j \le n$ and $i+j \le m$, where
each process enters the grid through the splitter at
$(0,0)$, and at each splitter $(i,j$) stops if it receives $\Stop$,
proceeds to $(i+1,j)$ if it receives $\Right$, and proceeds to
$(i,j+1)$ if it receives $\Down$.  Moir and Anderson show that 
in any execution in which $m$
processes follow this procedure, every process eventually receives
$\Stop$ at some splitter, before reaching one of the $2m$ output
wires.

\begin{figure}
\begin{center}
\includegraphics[scale=0.4]{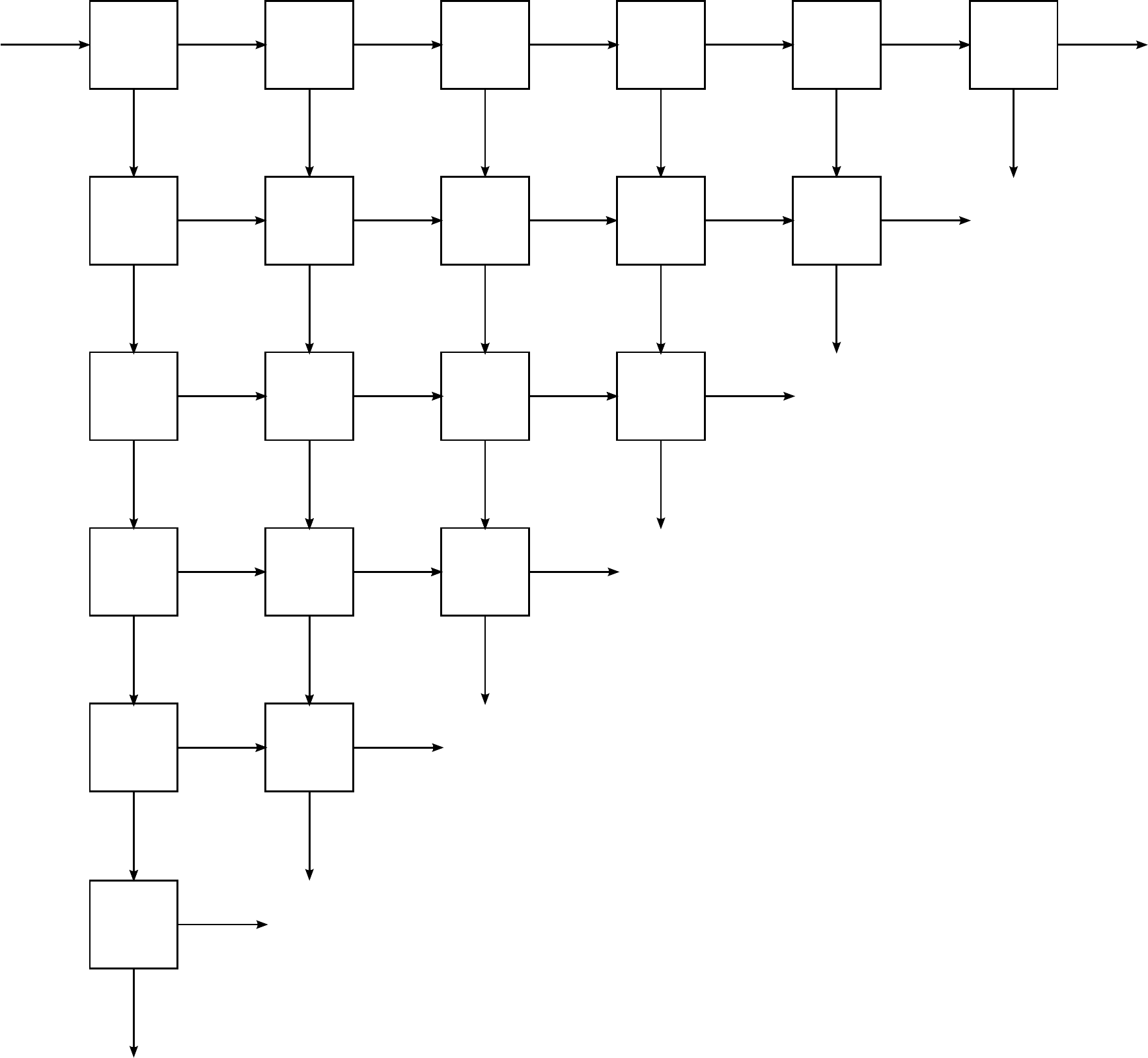}
\end{center}
\caption{A $6 \times 6$ Moir-Anderson grid.}
\label{figure-moir-anderson}
\end{figure}

\begin{figure}
\begin{center}
\includegraphics[scale=0.4]{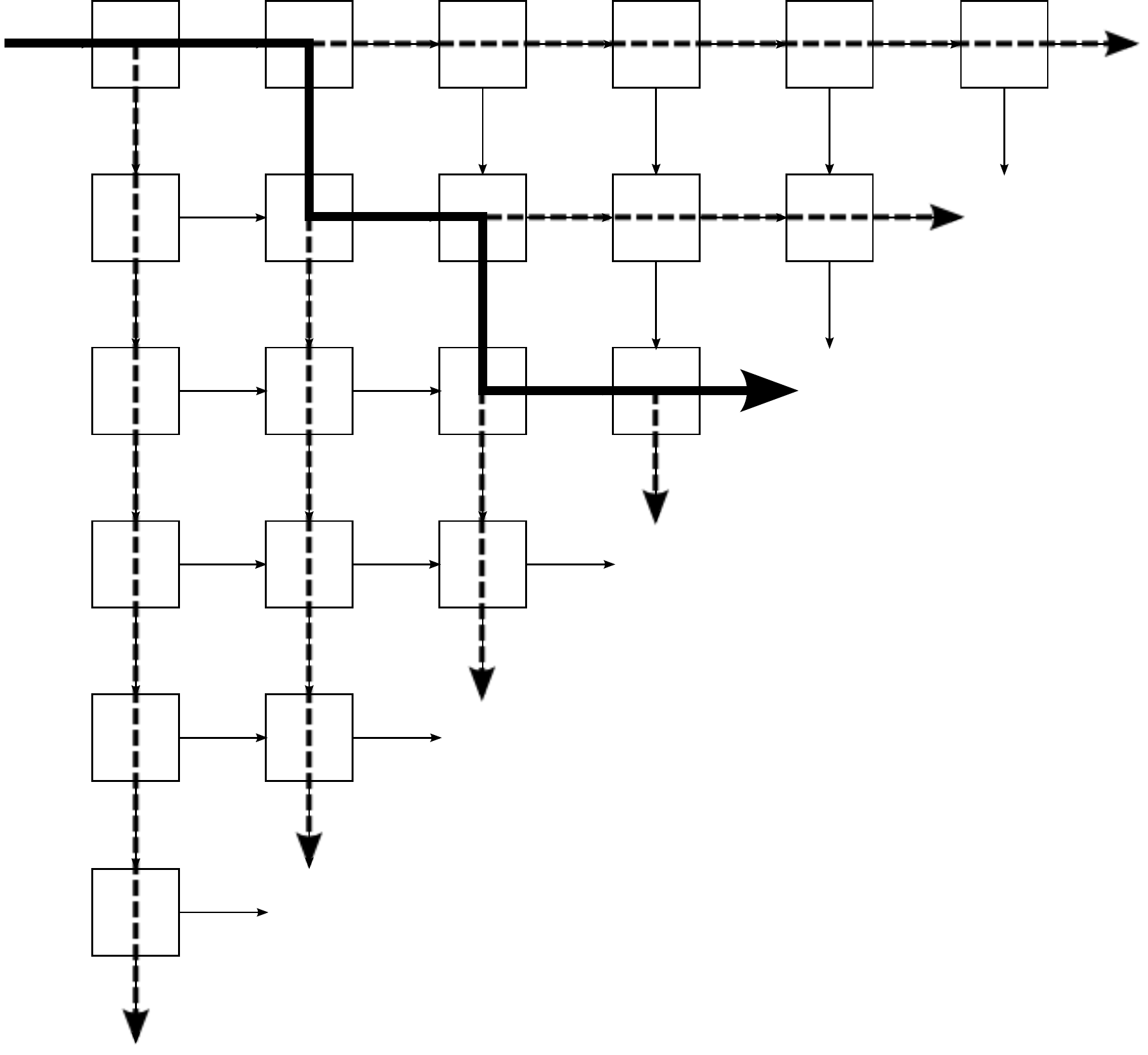}
\end{center}
\caption{Path taken by a single process through a $6 \times 6$
Moir-Anderson grid (heavy path), 
and the 6 disjoint paths it spawns (dashed paths).}
\label{figure-disjoint-paths}
\end{figure}

A simple explanation of this fact can be obtained by supposing that some
process $p$ reaches an output wire, and looking at the path it took to get
their (see Figure~\ref{figure-disjoint-paths}).  Each splitter on this
path must handle at least two processes (or $p$ would have stopped at
that splitter).  So some other process leaves on the other output
wire, either $\Right$ or $\Down$.  If we draw a path from each of
these wires that continues $\Right$ or $\Down$ to the end of the grid,
then along each of these $m$ disjoint paths either some splitter stops
a process, or some process reaches a final output wire, each of which
is at a distinct splitter.  
It follows that:

\begin{lemma}
\label{lemma-moir-anderson}
In an $m\times m$ Moir-Anderson grid, either all processes stop, or
\begin{align}
(\text{\# of nonempty output wires})
+ (\text{\# of stopped processes})
&\ge m+1,
\text{and}
\label{eq-wire-bound}
\\
(\text{\# of nonempty output splitters})
+ (\text{\# of stopped processes})
&\ge m.
\label{eq-splitter-bound}
\end{align}
\end{lemma}

An immediate corollary of the first bound \eqref{eq-wire-bound} 
is that an $m\times m$ Moir-Anderson grid stops
any set of $m$ or fewer processes, because otherwise there are not
enough to supply the $m+1$ processes in the inequality.
The second bound \eqref{eq-splitter-bound} will be useful in our
improved construction.

\section{Blockers}
\label{section-blockers}

Let an \concept{$(m,k)$-blocker} be a splitter network 
with the property that if $m$ processes
enter the network at its designated input gate, at least $k$ processes
stop somewhere in the network.  A single splitter is a $(1,1)$-blocker
(but is only an $(m,0)$-blocker for any $m > 1$).
An $m \times m$ Moir-Anderson grid is an $(m,m)$-blocker.

We will build an $(n,n)$-blocker out of $O(n^{3/2})$ splitters using
a sequence of $\sqrt{n}$ stages, each of which is an
$(n,\sqrt{n}\,)$-blocker.  After each stage,
all processes that have not stopped are fed into
the single input of the next blocker.  The overall structure is thus
similar to the cascaded-tree splitter networks considered
by~\cite{AttiyaKPWW2006}
but we obtain much lower space complexity
by using a combination of Moir-Anderson grids and binary trees 
in each stage instead of just a single
large binary tree.

The essential idea of each $(n,\sqrt{n}\,)$-blocker is to use a
Moir-Anderson grid of size $2 \sqrt{n}$ and apply inequality
\eqref{eq-splitter-bound} from
Lemma~\ref{lemma-moir-anderson} to show that at least $2 \sqrt{n}$
processes either stop inside the grid or leave the grid through
distinct output splitters.
Since there are only $n$ processes, fewer than $\sqrt{n}$ output
splitters   
will get more than $\sqrt{n}$ processes.  Deducting these overloaded
splitters from the $2 \sqrt{n}$ total
gives at least $\sqrt{n}$ output splitters that either (a) get between $1$ and
$\sqrt{n}$ processes, or (b) correspond to a stopped process inside
the grid.  By attaching a $(\sqrt{n},1)$-blocker to both outputs of
all $2\sqrt{n}$ splitters in the last layer, 
we stop at least one process for each splitter in class (a),
for a total of $\sqrt{n}$ stopped processes.

We have not yet shown how to build a $(\sqrt{n},1)$-blocker.  Here we
can just use a binary tree with $\sqrt{n}$ splitters:
\begin{lemma}
\label{lemma-binary-tree}
Any binary tree of $m$ splitters is an $(m,1)$-blocker.
\end{lemma}
\begin{proof}
By induction on $m$.  A single splitter is a $(1,1)$-blocker.  Given a
binary tree of $m$ splitters accessed by at least one process, either
the root node stops a process, or it sends at least one process to
each of its two subtrees.  Let $m_1$ and $m_2$ be the sizes of the two
subtrees; by the induction hypothesis, the subtrees are $(m_1,1)$ and
$(m_2,1)$ blockers, respectively.  So 
for no process to be blocked, we must send at least $m_1+1$ processes
to the first subtree and $m_2+1$ processes to the second, for a total
of $m_1+m_2+2 = m+1$ processes.  It follows that the full tree is an
$(m,1)$-blocker.
\end{proof}

Figure \ref{figure-blocker} shows an example of an
$(n,\sqrt{n}\,)$-blocker constructed in this way.
This uses
$\binom{2 \ceilsqrt{n}}{2} = \left(2+o\left(1\right)\right) n$ splitters for the
Moir-Anderson grid, plus $(2\ceilsqrt{n})\ceilsqrt{n} =
\left(2+o\left(1\right)\right)n$ splitters for the output blockers.
The depth of the blocker is $2\ceilsqrt{n} + \ceil{\,\lg n} =
\left(2+o\left(1\right)\right) \sqrt{n}$.  Summarizing:
\begin{lemma}
\label{lemma-blocker}
For any $n$, there is an $(n,\sqrt{n})$-blocker with
$(4+o(1))n$ splitters and depth $(2+o(1))\sqrt{n}$.
\end{lemma}

\begin{figure}
\begin{center}
\includegraphics[scale=0.6]{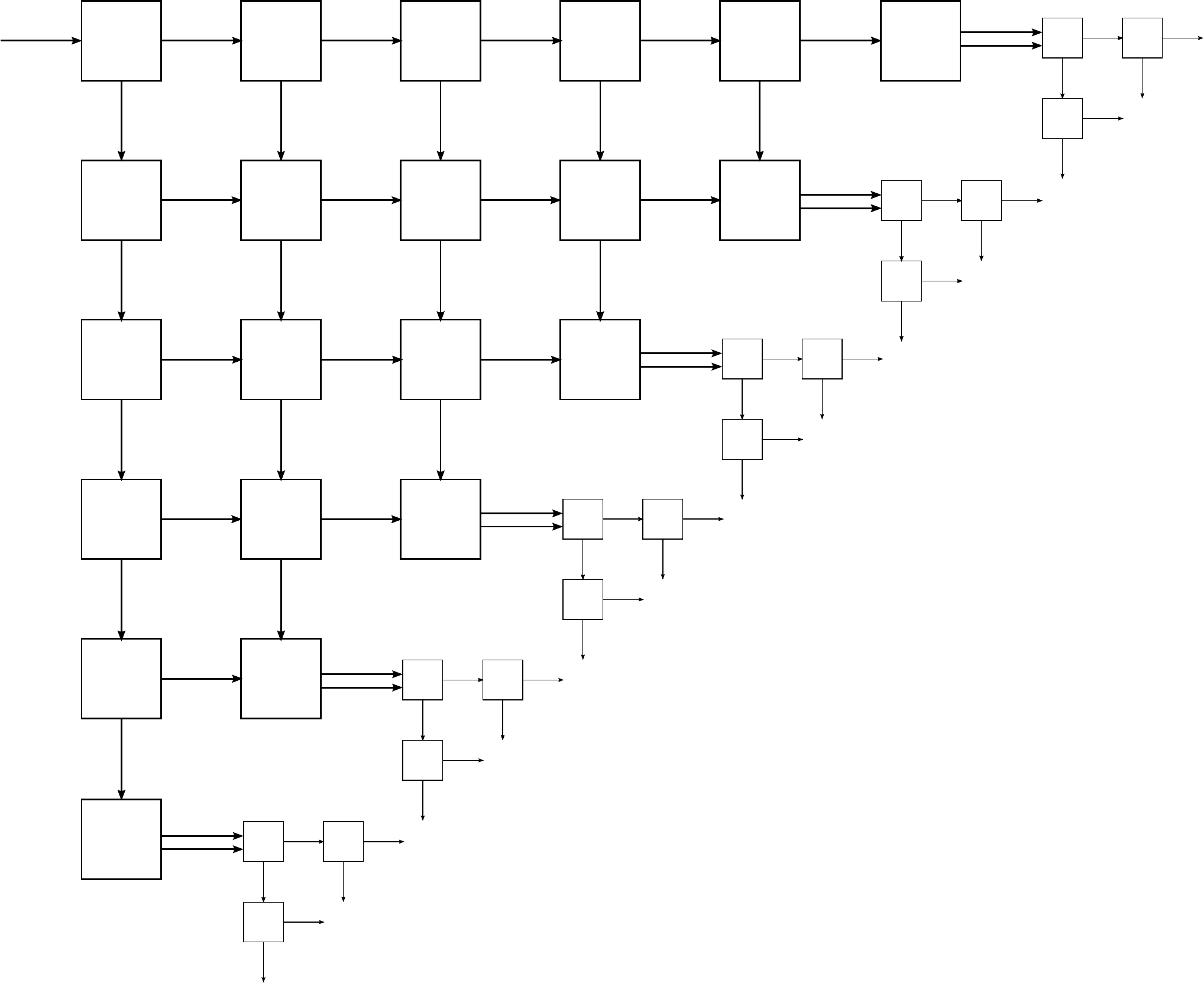}
\end{center}
\caption{A $(9,\sqrt{9})$-blocker, consisting of a
$6\times 6$ Moir-Anderson grid with a
$(3,1)$-blocker on each output wire, implemented as a binary
tree of splitters.}
\label{figure-blocker}
\end{figure}

\section{The full splitter network}
\label{section-the-full-splitter-network}

To obtain the full splitter network, we iterate our
$(n,\sqrt{n}\,)$-blocker $\sqrt{n}$ times.  Since each blocker stops
at least $\sqrt{n}$ processes, every process stops at some stage.
Summing the size and depth of the blockers over all $\sqrt{n}$
iterations gives:
\begin{theorem}
\label{theorem-full-network}
For any $n$, there is a network of $\left(4+o\left(1\right)\right)n^{3/2}$ splitters with depth
$\left(2+o\left(1\right)\right)n$ that
solves renaming deterministically for $n$ processes.
\end{theorem}

Though we have assumed a known, fixed bound on the number of processes
$n$, it is not hard to see that the algorithm could be made adaptive
by stringing together geometrically increasingly large networks, with
processes that fail to obtain a name in one network falling through to
the next.  This would give names in the range $O(k^{3/2})$ and time
complexity $O(k)$, where $k$ is the number of participating processes.

\subsection{Conclusions}
\label{section-conclusions}

We have shown that it is possible to build a splitter network that
assigns names to $n$ processes in $O(n)$ individual work using
$O(n^{3/2})$ space (and names).  
Because of the lower bound of Chlebus
and Kowalski~\cite{ChlebusK2008}, improving the time
complexity is not possible using a splitter network, but it may be
that further improvements to the structure of the network 
could reduce the space complexity.

\bibliographystyle{plain}
\bibliography{paper}

\end{document}